\documentclass{article}

\usepackage[scaled]{helvet} 
\usepackage{courier} 
\usepackage{eulervm}  
\normalfont
\usepackage[T1]{fontenc}
\usepackage{color}

\usepackage{mathrsfs}
\usepackage{hyperref}
\usepackage{amssymb,amsmath,amsthm}
\usepackage{bbold}

\newtheorem{Thm}{Theorem}[section]
\newtheorem{theorem}[Thm]{Theorem}
\newtheorem{proposition}[Thm]{Proposition}
\newtheorem{corollary}[Thm]{Corollary}
\newtheorem{lemma}[Thm]{Lemma}

\newtheorem{remark}{Remark}[section]
\newtheorem{remarks}{Remarks}[section]
\newtheorem{definition}[Thm]{Definition}

\begin{document}

\title{The Haag-Kastler Axioms for the  
$\mathscr{P}(\varphi)_2$ Model on the De Sitter Space}

\author{Christian D.\ J\"akel\footnote{jaekel@ime.usp.br, Universidade de S\~ao Paulo (USP), Brasil} \ 
and Jens Mund\footnote{mund@fisica.ufjf.br, Departamento de Fisica, Universidade de Juiz de Fora, Brasil}}        

\maketitle

\begin{abstract}
We establish the Haag-Kastler axioms for a class of interacting quantum field theories on the
two-dimensional de Sitter space, which satisfy \emph{finite speed of light}. 
The ${\mathscr P} (\varphi)_2$ model constructed in \cite{BJM-1}, describing 
massive scalar bosons with polynomial interactions, provides an example.  
\end{abstract}

\section{Introduction}
\label{sec:1}
 
The construction of interacting relativistic quantum field theories in $1+1$ dimensional Minkowski space 
exhibiting particle production, due to Glimm and Jaffe \cite{GJ}, is one of the crown jewels  
of mathematical physics. An interesting aspect of their construction is that the Haag-Kastler net 
for the interacting quantum field can be formulated on the Fock space for the free massive scalar field.
In fact, the von Neumann algebras associated to double cones whose base lie on the time-zero 
surface (\emph{i.e.}, the Cauchy surface) can be identified in both models. The spatial translations can 
as well be carried over from the model representing non-interacting particles. (Inspecting the Lie algebra 
relations of the Poincar\'e group, one realises that the Lorentz boosts can \emph{not} be taken over from the free
theory when the time evolution is modified. In other words: a deformation of the time-evolution requires 
a modification of the Lorentz boosts, while the spatial translations may remain unchanged.) Due to Haag's 
theorem (see \cite{Weiner} for a formulation in an operator-algebraic language), the interacting 
vacuum state can \emph{not} be given by a vector in Fock space (it is not normal with respect to the representation 
of the free massive field), but as long as the localization region of the observables one is interested in is bounded 
(in space-time) there are vectors in Fock space, which implement the interacting vacuum state locally. 
In other words, the interacting vacuum state is locally Fock.  

In de Sitter space, the situation is more favourable in two respects. Firstly, 
as the Cauchy surface is compact, Haag's no-go theorem no longer applies (see again \cite{Weiner})
and any state which is locally Fock is normal with respect to the Fock representation.  
Secondly, if one replaces the \emph{spectrum condition} (used in Minkowski space) by the 
\emph{geodesic KMS condition} of Borchers and Buchholz~\cite{BoB} to ensure stability of
the de Sitter vacuum state, the de Sitter vacuum is a thermal state for the Hawking temperature $(2\pi r)^{-1}$
for the unitary group of Lorentz boosts associated with any given wedge region  (see \cite{BoB}). 
This characterization of de Sitter vacua applies both to interacting and  non-interacting theories. 

Once one has constructed the family of von Neumann algebras describing the free massive scalar field, 
the facts just presented suggest that a look\footnote{In fact, the authors
would argue that there is a deep connection between Tomita-Takesaki modular theory and the 
representation theory of the group $SL(2, \mathbb{R})$ (and, more general, non-compact Lie-algebras which 
contain $SL(2, \mathbb{R})$ as a subgroup).} at (a generalisation \cite{DJP} of) Araki's  perturbation theory of KMS states
might be very useful: one  can pick an arbitrary wedge (which we denote as $W_1$)  and select a cyclic and 
separating vector $\Omega$ in the natural positive cone associated to the von Neumann algebra $\mathcal{A}_\circ(W_1)$ 
for the wedge $W_1$ (of the free massive field) and the Fock vacuum vector $\Omega_\circ$. 
If this vector happens to be invariant under the rotations for the model representing 
non-interacting particles (leaving  some geodesic Cauchy surface, which contains the edges of the wedge~$W_1$, 
invariant), one may hope that the candidates for the new boosts (\emph{i.e.}, implemented by the modular group for the 
pair~$(\mathcal{A}_\circ(W_1), \Omega)$) together with the old rotations generate a new 
representation of $SO_0(1,2)$. General criteria, which ensure that this is indeed the case, 
will be investigated elsewhere. However, for interesting examples, it is already known that the group closes:  
together with J.~Barata the authors have shown in \cite{BJM-1} that the interacting 
de Sitter vacuum vector for the ${\mathscr P} (\varphi)_2$  
model (which lies in the natural positive cone  for the pair $(\mathcal{A}_\circ(W_1), \Omega_\circ)$) 
can be used to define interacting boosts which, together with the 
rotations for the free field, generate a \emph{new} unitary representation of the Lorentz group $SO_0(1, 2)$.  

Now, if one already has an \emph{interacting} representation of the Lorentz group acting on Fock space, it still remains to be 
shown that  this new representation gives rise to a new family of von Neumann algebras satisfying the Haag-Kastler axioms. 
Our basic strategy to establish this fact is to identify the von Neumann algebras associated to the wedge~$W_1$ in the 
free and in the interacting theory. All other local algebras are then defined 
by applying the new representation of $SO_0(1, 2)$ (which represents the interacting dynamics) and taking intersections. 
This machinery ensures most of the Haag-Kastler axioms, but non-triviality of the double cone algebras\footnote{We note that 
non-triviality of the local algebras is difficult to establish in the approach pioneered by Gandalf Lechner, see \cite{Le2, Le3}.}
has to be established separately. As the rotations are carried over from the model representing non-interacting particles, the von 
Neumann algebras associated to wedges whose bases lie on the time-zero surface (\emph{i.e.}, the Cauchy surface) can be 
identified in both models, too. The same argument does not apply to the von Neumann algebras associated to double cones:
a von Neumann algebra associated to a double cone $\mathcal{O}$, whose base lies on the geodesic Cauchy surface, 
must be contained in the intersection of the von Neumann algebras associated to \emph{all} wedges which
contain~$\mathcal{O}$ (and \emph{not} only those whose base lie on the Cauchy surface). 
Consequently, an additional condition is needed to ensure non-triviality of the double cone algebras. A sufficient condition is 
that the interacting representation satisfies the \emph{finite speed of light} property introduced by Glimm and Jaffe~\cite{GJ}. 
As we will show in this work, the latter ensures that, just like in the Minkowski space case analyzed by Glimm and Jaffe, the 
von Neumann algebras associated to double cones whose base lie on the time-zero surface (\emph{i.e.}, the Cauchy surface) 
can be identified in both models. Finite speed of light  has been verified for the ${\mathscr P} (\varphi)_2$ model on de Sitter 
space in~\cite{BJM-1}.

\section{One-particle space}
\label{sec:2}

The two-dimensional de Sitter space 
	\begin{equation}
		\label{eqdSMin}
		\mathbb{dS} \doteq \left\{  
		x   \in \mathbb{R}^{1+2} 
		\mid 
		x_{0}^{2} - x_{1}^{2} - x_{2}^{2} = - r^2 \right\} \; , 
		\quad r >0 \; ,
	\end{equation}
can be viewed as a one-sheeted \emph{hyperboloid}, embedded in the $(1+2)$-dimensional Minkowski 
space~$\mathbb{R}^{1+2}$. The embedding (\ref{eqdSMin}) is compatible with the metric and the causal 
structure, \emph{i.e.}, the de Sitter space $\mathbb{dS}$ inherits its metric and the causal structure from the 
ambient Minkowski space. 

The isometry group of $\mathbb{dS}$ is the Lorentz group $O(1,2)$. The connected component  
containing the identity is $SO_0(1,2)$. This subgroup is generated by the \emph{rotations} 
	\[
		R_0(\alpha) \; \doteq \; \begin{pmatrix}
					1 &  0 &0 \cr
					0 &  \cos \alpha &  - \sin \alpha  \cr
					0  & \sin \alpha & \cos \alpha 
					\end{pmatrix} \; ,
		\qquad \alpha\in [0, 2 \pi) \; , 
	\]
and the \emph{Lorentz boosts} 
	\[
		\Lambda_1 (t)   \; \doteq \;   \begin{pmatrix}
				 \cosh t  &  0 &\sinh t \cr
     					  0  &  1 & 0  \cr
 				  \sinh t &  0 & \cosh t \end{pmatrix} \; ,							
		\qquad 
				t \in \mathbb{R} \; .   
	\]
We will also need the rotated boosts  
		\begin{equation} 
			\label{Lambdaalpha}
				\Lambda^{(\alpha)} (t)= 
				R_{0}(\alpha) \Lambda_{1}(t) R_{0}(-\alpha) \; , 
				\qquad t \in \mathbb{R} \; . 
		\end{equation} 
According to our convention, the boosts $\Lambda_1 (t) \equiv \Lambda^{(0)} (t)$ keep 
the $x_1$-axis  invariant, and therefore correspond  to  boosts in the $x_2$-direction. 

\goodbreak
The circle 
	\[
		S^1  \doteq \bigl\{  x   \in \mathbb{dS} \mid x_{0} = 0 \bigr\}
	\]
forms a \emph{Cauchy surface} for $\mathbb{dS}$. For two points $x= (0, r \sin \psi, r \cos \psi)$ 
and $y= (0, r \sin \psi', r \cos \psi')$ on the circle $S^1$, the Wightman two-point function of a scalar free field 
(analysed in \cite{BM}) equals 
	\begin{equation}
	\label{2-point}
		{\mathcal W}_\nu^{(2)} (  x,  y) =  c_\nu \, P_{s^+}  ( - \cos ( \psi -\psi'))  \; . 
	\end{equation}
Here $P_{s^+} $ is the \emph{Legendre function} for the parameter 
	\begin{equation} 
		\label{dd1} 
			s^\pm= -\tfrac{1}{2}  \mp i \nu \; , 
		\end{equation} 
with
	\begin{equation} 
		\label{dd-1} 
	\nu =  
			\begin{cases}
				i \sqrt{\frac{1}{4} -\zeta^2} & \text{if $ 0< \zeta < 1 /2$}  \, ,\\
				 \sqrt{\zeta^2 - \frac{1}{4} } & \text{if $ \zeta \ge 1 /2$} \, ,
			\end{cases} 
	\end{equation} 
and $\zeta$ the eigenvalue of the Casimir operator of $SO_0(1,2)$.
The two-point function \eqref{2-point} gives rise to a scalar product on $C^\infty (S^1)$:
	\begin{align}
				\label{eq:def-scalar-product-2}
		\langle h , h' \rangle_{\mathcal{H} }  & \doteq c_\nu
				\int_{S^1} r \, {\rm d}\psi \int_{S^1} r \, {\rm d}\psi' \; \overline{h(\psi)} \, 
				\\
				& \qquad \qquad \qquad \times P_{s^+}\big( - \cos(\psi - \psi') \big) \, h'(\psi') \; .
				\nonumber
	\end{align}
The value of the positive normalisation constant $c_\nu$ is
	\[
		c_\nu = -\frac{1}{2\sin  (\pi s^+ )} = 
		\frac{1}{2 \cos (i \nu \pi ) }  \; . 
	\] 
Note that the singularity for $\psi = \psi' $ is integrable; see for instance p.~364 in~\cite{BM}. 
The completion of $C^\infty (S^1)$ w.r.t.~this scalar product is a one-particle 
Hilbert space, which we denote by $\mathcal{H}$. Clearly, the scalar product \eqref{eq:def-scalar-product-2}
depends on the value of the Casimir operator $\zeta >0$, but we will suppress this dependence in the notation. 

\bigskip
The Fourier coefficients of the Legendre function appearing in \eqref{2-point} and 
\eqref{eq:def-scalar-product-2} were computed in the proof of Proposition 4.7.3 in \cite{BM}. 
The result can be casted in the following form: 

\begin{lemma} 
The scalar product (\ref{eq:def-scalar-product-2}) can be 
expressed as 
	\[
		\langle h , h' \rangle_{\mathcal{H} }  = 
		\bigl\langle h  , \tfrac{1}{2\omega} 
		h' \bigr\rangle_{L^2( S^1, r {\rm d} \psi)} \; ,  
	\]
with $\omega $ a strictly positive self-adjoint operator on $L^2( S^1, r {\rm d} \psi)$ with Fourier coefficients 
	\begin{equation}
	\label{omega-k}
		\widetilde {\omega}(k) = {r}^{-1} \, (k+s^+)
					 \frac{\Gamma \left( \frac{k+s^+}{2} \right)}{ \Gamma \left( \frac{k-s^+}{2} \right)}
			\frac{ \Gamma \left( \frac{k+1-s^+}{2} \right) }{ \Gamma \left( \frac{k+1+s^+}{2} \right)} \; ,
			\; \;  k\in \mathbb{Z} \; .
	\end{equation}
\end{lemma}

We note that $\widetilde {\omega}(0)>0$ and $\widetilde {\omega}(-k) = \widetilde {\omega}(k)$. 
The map $\mathbb{R}^+ \ni k \mapsto \widetilde {\omega}(k)$
is monotonically increasing and, asymptotically, one finds
	\[
		\lim_{ k \to \pm \infty} \frac{ r \widetilde {\omega}(k)}{ | k |} = 1 \; , 
	\]
just like for the free  massive scalar field on Minkowski space. 

\bigskip

One of the key results in \cite[Theorem~4.7.5]{BJM-1} is that $\mathcal{H}$
carries a representation of the Lorentz group: 

\begin{theorem} 
\label{UIR-S1}
The rotations 
	\[
	\bigr( u(R_0(\alpha)) h \bigl) (\psi) = h (\psi - \alpha) \; , 
	\quad \alpha \in [0, 2\pi) \; , \;  h \in \mathcal{H}  \; , 
	\]
and the boosts 
\label{Umhatpage}
	\begin{equation}
	\label{UIR}
		u (\Lambda_1(t)) = {\rm e}^{i t \omega  r \, \widehat{\cos}} \; , \qquad t \in \mathbb{R} \; , 
	\end{equation}
generate a unitary irreducible representation of $SO_0(1,2)$ on $\mathcal{H}$, 
which extends to an (anti-)unitary representation of $O(1,2)$.
In particular, the reflection $\Theta_{W_1} := P_1T$ at the edge of the wedge $W_1$ is represented by the anti-linear map
	\[ 
		u \bigl(\Theta_{W_1}\bigr) h = C {P_1}_* h \; , \qquad h  \in \mathcal{H} \; , 
	\]
where ${P_1}_* h (\psi) \doteq h (\pi - \psi)$ and $(Ch)(\psi)\doteq \overline{h(\psi)}$ denotes the complex conjugation. 
\end{theorem}

\begin{remarks}
\quad 
\begin{itemize}
\item[$i.)$]
The symbol  "$\widehat{\cos}$" in the exponent in \eqref{UIR} denotes the 
multiplication operator, mapping a function $f (\psi) $ to  $\cos \psi \cdot f (\psi)$.
\item[$ii.)$] It is remarkable that both the representations of the \emph{principle series} (those with eigenvalue $\zeta \ge 1/2$
of the Casimir operator) as well as the representations of the \emph{complementary series} (those with $0< \zeta < 1/2$) 
can be casted in the form given in Theorem~\ref{UIR-S1}. Note that both $\mathcal{H}$ and $\omega$ depend on $\zeta \, $.
\end{itemize}
\end{remarks}

From the given unitary representation of the Lorentz group we now define a family of $\mathbb{R}$-linear 
subspaces of $\mathcal{H}$ by a standard construction, called \emph{modular localization} (see~\cite{BGL,S97}), 
which draws inspiration from the results \cite{BiWia,BiWib} of Bisognano and Wichmann.  

\begin{definition}(Modular Localization).
\label{def-h-loc}
\begin{itemize}
\item[$i.)$] For the wedge $W _1\doteq \bigl\{  x \in \mathbb{dS} \mid x_2 > |x_0 | \bigr\}$, we set 
	\[			
	\mathcal{H} ( W _1 ) \doteq
		 \bigl\{ h \in {\mathscr D} \bigl( u (\Lambda_1(i \pi r ))\bigr) \mid u(P_1T) u \bigl(\Lambda_1(i \pi r )\bigr) h = h \bigr\} \; . 
	\]
\item[$ii.)$] For an arbitrary wedge $W= \Lambda W_1$, $\Lambda \in SO_0(1,2)$,
we set 
	\begin{equation} \label{h-tilde-of-W}
		\mathcal{H} ( W ) \doteq u(\Lambda) \mathcal{H} ( W_1) \; . 
	\end{equation}
\item[$iii.)$] For a causally complete, open and bounded region ${\mathcal O}$, we set
	\begin{equation}
	\label{h-tilde-of-O}
		\mathcal{H}({\mathcal O}) \doteq \bigcap_{ {\mathcal O} \subset W} \mathcal{H}(W) \; . 
	\end{equation}
\end{itemize}
\end{definition}

Note that $\mathcal{H} ( W )$ is well-defined by \eqref{h-tilde-of-W} due to the following standard
argument~\cite{BGL}: The only Lorentz transformations, which leave the wedge $W_1$ invariant, 
are of the form $\Lambda \equiv \Lambda_1 (t)$ for some $t \in \mathbb{R}$. But the representer 
$u\bigl( \Lambda_1(t) \bigr)$ of such a boost commutes with both $u(P_1T)$ and with 
$u \bigl(\Lambda_1(i \pi r )\bigr)$ and thus leaves the $\mathbb{R}$-linear 
subspace $\mathcal{H} ( W _1 )$ invariant.  

\begin{proposition} 
\label{Prop-ii.4}
The subspaces introduced in Definition~\ref{def-h-loc} have the following properties:
\begin{itemize}
\item[$i.)$] (Wedge Duality). The $\mathbb{R}$-linear subspace $\mathcal{H}(W')$ for the opposite wedge
	\[ 
		\qquad 
		W' \doteq \bigl\{  x \in \mathbb{dS} \mid x \; \hbox{space-like separated from} \; W \bigr\}  
	\]
equals the symplectic complement  
	\[
		\qquad 
		\mathcal{H} ( W )'  
		\doteq
		\bigl\{ h \in \mathcal{H} \mid  \Im \langle h, g \rangle = 0 \; \; 
		\forall g \in \mathcal{H} ( W )  \bigr\}   
	\]
of $\mathcal{H} ( W )$.  
\item [$ii.)$] (Covariance). For $\Lambda \in SO_0(1,2)$ and $\mathcal{O}$ a causally complete, open, 
connected and bounded region,  
	\[
		\qquad
		\mathcal{H} ( \Lambda \mathcal{O}) = u (\Lambda) \mathcal{H} ( \mathcal{O}) \, .
	\]
\item [$iii.)$] (Microcausality). For two space-like separated causally complete, open, bounded regions
$\mathcal{O}_1$ and $\mathcal{O}_2$,  
	\[
		\qquad
		\Im \langle h_1, h_2 \rangle_{\mathcal{H}} = 0 
		\quad \forall h_i \in \mathcal{H} ( \mathcal{O}_i) \; , \quad i=1,2 \, .
	\]
\end{itemize}
\end{proposition}

\begin{proof}
The items $i.)$ and $ii.)$ follow from Proposition 5.2 of~\cite{BGL}, while $iii.)$ follows from Theorem 5.4
of~\cite{BGL}. For the convenience of the reader, we indicate the arguments.
As a consequence of the group relations, the anti-linear operator 
	\[
		s_{W_1} \doteq u \bigl(\Theta_{W_1}\bigr) u \bigl(\Lambda_1(i \pi r)\bigr)
	\]
is the Tomita operator for $\mathcal{H}(W_1)$;
namely, it is a densely defined involution (\emph{i.e.}, $ s_{W_1}^2\subset \mathbb{1}$) 
and its eigenspace for eigenvalue $+1$ is just $\mathcal{H}(W_1)$: 
from the group relations between $\Lambda_1(t)$, $R_0(\pi)$ and $\Theta_{W_1}= P_1T$  
it follows that the operator $u(R_0(\pi))u (\Theta_{W_1})$ commutes with the anti-linear operator 
$s_{W_1}$ and thus leaves the $\mathbb{R}$-linear subspace $\mathcal{H}(W_1)$
invariant. But this implies that 
	\begin{equation} 
		\label{eqCPT}
		u \bigl(\Theta_{W_1}\bigr)  \mathcal{H}(W_1) = \mathcal{H}(W_1') \; , 
	\end{equation}
since $R_0(\pi) W_1=W_1'$. On the other hand, a general result on Tomita operators~\cite[Prop.~2.3]{RvD} 
asserts that the anti-unitary part $u (\Theta_{W_1})$ in the polar decomposition of $s_{W_1}$ maps 
the eigenspace for the eigenvalue $1$, namely $\mathcal{H}(W_1)$, 
onto its symplectic complement $\mathcal{H}(W_1)'$. Thus, Eq.~\eqref{eqCPT} is just \emph{wedge duality} 
(\emph{i.e.}, property $i.)$ in the proposition) for $W_1$.  
For other wedges, duality follows from wedge covariance~\eqref{h-tilde-of-W}. 

To prove property $iii.)$, \emph{microcausality}, pick a wedge $W$ such that $\mathcal{O}_1 \subset W$ and $\mathcal{O}_2 \subset W'$. 
Then
	\[
		\mathcal{H}(\mathcal{O}_1)\subset \mathcal{H}(W)  = \mathcal{H}(W')' \subset \mathcal{H}(\mathcal{O}_2)' , 
	\]
where we have used wedge duality.

Property $ii.)$, \emph{covariance}, follows from the definitions of $\mathcal{H}(W)$ and $\mathcal{H}(\mathcal{O})$, respectively.
\end{proof}

It will be useful to have an explicit formula for the real subspaces associated to a certain class of regions, 
namely double cones or wedges with \emph{base} on~$S^1$.
For these regions, an alternative localization map (widely used in quantum field theory on Minkowski space) is available, which exploits the 
support properties of Cauchy data of solutions of the Klein-Gordon equation, by identifying  
$\Re h$ and $\omega^{-1}\Im h$, for $h\in \mathcal{H}$, with initial data of a 
solution of the Klein-Gordon equation: 
for $I$ an open interval in $S^1$, one defines the $\mathbb{R}$-linear subspaces 
	\begin{equation}
		\label{H_I}
		\mathcal{H}_{I}  
		\doteq 
		\left\{ h \in \mathcal{H} \mid  
		\operatorname{supp} \Re h  \subset I \, , \;
                \operatorname{supp} \omega^{-1}\Im h \subset I
                \right\} .
         \end{equation}
It has been shown in~\cite[Prop.~6.5.5]{BJM-1} that 
	\begin{equation}
		\label{finite-speed-of-light-1particle}
		u(\Lambda_1(t)) \mathcal{H}_{I} \subset \mathcal{H}_{I_t} \; ,
		\qquad I_t\doteq \Gamma(\Lambda_1(t)I)\cap S^1 \; ,
\end{equation}
where  $\Gamma(M)$ is the domain of dependence of a set~$M$, \emph{i.e.}, 
the union of the future~$\Gamma^+ (M)$ and the past $\Gamma^-(M)$ of $M$. Equation
\eqref{finite-speed-of-light-1particle} expresses the hyperbolic character of the Klein-Gordon 
equation in the Hilbert space context. It is common usage to refer to \eqref{finite-speed-of-light-1particle}  
as \emph{finite speed of light}. A similar condition, which applies to interacting theories, will 
be presented in \eqref{e6.1ef}. 

\begin{proposition} 
The subspaces introduced in Definition~\ref{def-h-loc} have the following properties:
\begin{itemize}
\item[$i.)$] (Modular Localization $\Leftrightarrow$ Localization of Cauchy data). 
For $I$ a bounded open interval of length $ | I |  \le \pi \, r$ in $S^1$ there holds 
	\begin{equation}
	\label{cauchy-localization}
		\qquad
		\mathcal{H} ( \mathcal{O}_I) 
		= \mathcal{H}_I \;  , 
	\end{equation}
where $\mathcal{O}_I = I''$ denotes the \emph{causal completion} of the interval $I$ in $\mathbb{dS}$. 
\item [$ii.)$] (Additivity). For any open interval $I\subset I_+$, we
have 
	\begin{equation}
	\label{additivity}
		\mathcal{H} (W_1) =  \bigvee_{R_0(\alpha)I\subset I_+}
                        \mathcal{H} \bigl( \mathcal{O}_{R_0(\alpha)I} \bigr)  \; ,  
	\end{equation}
        where $\vee$ denotes the closure of the $\mathbb{R}$-linear span in $\mathcal{H}$. 
\item [$iii.)$] (Standard Subspaces). 
The $\mathbb{R}$-linear subspaces $\mathcal{H}(\mathcal{O})$ are \emph{standard}, \emph{i.e.}, 
	\begin{equation}
	\mathcal{H}(\mathcal{O}) \cap i \mathcal{H}(\mathcal{O}) = \{0 \}
	\; , 
	\qquad
	 \overline{\mathcal{H}(\mathcal{O}) + i \mathcal{H}(\mathcal{O})} = \mathcal{H}  \; ,  
	\label{standard property}
	\end{equation}
for all double cones $\mathcal{O} \subset \mathbb{dS}$. 
\end{itemize}
\end{proposition}

\begin{proof} 
We prove property $i.)$ first for the wedge $W_1$. Let $I_+$ be the open
half-circle $I_+ \doteq \{  x   \in S^1 \mid x_{2} >0 \}$. It has been shown in~\cite[Prop.~6.4.3]{BJM-1} that 
$\mathcal{H}_{I_+}$ is contained in $\mathcal{H}(W_1)$; see property $iv.)$ of Definition~A.7 in~\cite{BJM-1}.  
But $\mathcal{H}_{I_+}$ is invariant under the modular unitary group associated with $\mathcal{H}(W_1)$, 
namely $u(\Lambda_1(t))$, by the finite speed of light property~\eqref{finite-speed-of-light-1particle}. Further,
$\mathcal{H}_{I_+}$ is also a standard subspace. Takesaki's Theorem on standard subspaces 
(see,  \emph{e.g.}, \cite{LL}) then asserts that
	\begin{equation} 
		\label{HIHW}
		\mathcal{H}_{I_+}= \mathcal{H}(W_1) \; ,
	\end{equation}
as claimed. An alternative proof of \eqref{HIHW}  is provided by combining~\cite[Prop.~6.4.5]{BJM-1} (which 
provides a unitary map between the covariant one-particle space and the canonical one) 
with~\cite[Theorem~6.4.6]{BJM-1}, which discusses restrictions of this unitary map to localised regions.

Let now $I$ be an interval as in the proposition. As $\mathcal{O}_I$
is causally complete,  
	\[
		\bigcap_{ \mathcal{O}_I \subset W} W = \mathcal{O}_I = W (\alpha) \cap W (\beta)  
	\]
for some fixed $\alpha, \beta \in [0, 2\pi)$, where  
	\[
		W (\alpha) \doteq R_0(\alpha) W_1 \; , \quad \alpha
                \in [0, 2 \pi) \;  
	\]
denotes a wedge whose edges lies on $S^1$. Inspecting the definitions and \eqref{HIHW}, 
we find that   
	   \[
		\mathcal{H}_I  
		= \mathcal{H}_{R_0(\alpha) I_+} \cap
			\mathcal{H}_{R_0(\beta) I_+}
		= \mathcal{H} \bigl( W (\alpha)  \bigr) \cap
			\mathcal{H} \bigl(  W (\beta)  \bigr)  \; .  
	\]
As both $W (\alpha)$ and $W (\beta)$ are wedges which contain $\mathcal{O}_I$, we have 
	\[
		\mathcal{H} ( \mathcal{O}_I) \subseteq \mathcal{H}_I \; . 
	\]
Next, we assume that $W$ is an arbitrary wedge which contains~$\mathcal{O}_I$. The opposite wedge $W'$ 
of $W$ is, like any wedge, of the form $\Lambda^{(\beta)}(t) R_0 (\alpha) W_1$ for suitable $\alpha, \beta$ and~$t$.
As a consequence of finite speed of light\footnote{This 
property has been shown in~\cite[Prop.\ 6.5.5]{BJM-1} for $\Lambda_1(t)$, but also holds for $\Lambda^{(\beta)}(t)$ as
$u\big(R_0(\alpha)\big) \;  \mathcal{H}_I =  \left\{ h \in \mathcal{H} \mid
\operatorname{supp} \Re h \subset R_0(\alpha)I \, ,  \; \operatorname{supp}
\omega^{-1}\Im h \subset R_0(\alpha) I \right\}$ for all $\alpha \in [0, 2\pi)$.} \eqref{finite-speed-of-light-1particle} 
we have   
	\begin{align*}  
	  \mathcal{H} ( W') & =  u\bigl( \Lambda^{(\beta 
            )}(t)\bigr)
	  \mathcal{H}\bigl(W( \alpha )\bigr)
         \subset  \mathcal{H}_J  
        \end{align*}
with $J  = \Gamma\big(  \Lambda^{(\beta)}(t)R_0(\alpha) I_+ \big) \cap S^1 $,
where $\Gamma(M)$ is the
domain of dependence of a set~$M$, \emph{i.e.}, 
the union of the future $\Gamma^+ (M)$ and the past 
$\Gamma^-(M)$ of $M$. Note that 
\begin{itemize}
\item [$a.)$] $\Gamma\big( \Lambda^{(\beta )}(t)   R_0(\alpha) I_+ \big) = \Gamma (W') $; 
\item [$b.)$] $W'$ is space-like to $I$, since $W$ contains~$\mathcal{O}_I$. 
\end{itemize}
Hence $\Gamma(W')\cap S^1$ is in the interior $I^c \doteq S^1 \setminus \overline{I} $
of the complement of $I$ within $S^1$, and the same holds for $J$. Thus
	\[
		\mathcal{H} (  W' ) \subset \mathcal{H}_{I^c}.
	\]
Wedge duality now implies
	\[
		\mathcal{H} (  W ) = \mathcal{H} ( W' )'  
		 \supseteq  \mathcal{H}_I \; .  
        \]
This verifies \eqref{cauchy-localization}. 

Property $ii.)$. Additivity follows from the covariant formulation of the one-particle Hilbert space,
which uses the Fourier-Helgason transformation to define $\mathbb{R}$-linear subspaces 
associated to bounded space-time regions, see~\cite{BJM-1}. However, it can also 
be verified directly, exploring ideas of \cite[Sect. 5]{A2}. Using~\eqref{cauchy-localization}, 
the statement \eqref{additivity} is equivalent to 
the following one: for any open interval $I\subset I_+$,  
	\begin{equation}
	\label{additivity-2}
		\mathcal{H}_{I_+} =  \bigvee_{R_0(\alpha)I\subset I_+}
                        \mathcal{H}_{R_0(\alpha)I}   \; ,  
	\end{equation}
where $\vee$ denotes the closure of the $\mathbb{R}$-linear span in $\mathcal{H}$. 
Since the closure of~$I_+$ is compact, it is sufficient to prove additivity for two overlapping open 
intervals $I_1, I_2 \subset I\equiv I_1 \cup I_2$. Inspecting \eqref{H_I}, we see that we  
have to show that any $h \in \mathcal{H}_{I}$ can approximated by a sequence of functions  
$\{ h_n \}_{n \in \mathbb{N}}$ such that 
	\[
			h_n = h_n^{(1)} + h_n^{(2)} \;, 
	\]
with
	\[
			\operatorname{supp} \Re h_n^{(i)}  \subset I_i \; , 
			\quad
			\operatorname{supp} \Im \omega^{-1} h_n^{(i)}  \subset I_i \; , 
			\quad i = 1,2 \; . 
	\]
Since $L^2(I, r {\rm d} \psi)$ 
is dense in $\mathcal{H}_{I}$, we can choose $\Re h_n^{(i)} \in L^2(I_i, r {\rm d} \psi)$. Additivity for the 
real part is then a consequence of the additivity of the relevant $L^2$-spaces. Since $\omega^{-1}$ is 
an injective bounded operator on $\mathcal{H}_{I}$, it is bijective onto its image, and thus every 
$h \in \mathcal{H}_{I}$ is of the form 
	\[
		h = \omega ( \omega^{-1} h ) \; , 
	\]
where $g:=\omega^{-1} h \in \mathbb{H}^{1/2} (S^1)$ has support in $I$. 
Here $\mathbb{H}^{1/2} (S^1) \subset L^2(S^1, r {\rm d} \psi)$ is the closure 
of $C^\infty(S^1)$ with respect to the norm
	\[
			\| f \|_{\mathbb{H}^{1/2} (S^1)} := \sum_{k \in \mathbb{Z}} \widetilde{\omega} (k) | f_k |^2 \; . 
	\]
Note that the map $\omega \colon \mathbb{H}^{1/2} (S^1) \to \mathcal{H}$ is unitary. 
Now, any $g$ in $L^2(I, r {\rm d} \psi)$ can be decomposed in the form 
	\[
		 g = \chi_1 g + \chi_2 g \; , \qquad \chi_i \in C^\infty_0 (I_i) \; , \quad i =1,2 \; . 
	\]
The support properties are clearly satisfied. Hence additivity of the imaginary part $\Im h$ of $h$ follows, once
one has verified that the multiplication by a $C^\infty_0$-function defines a bounded operator 
on~$\mathbb{H}^{1/2} (S^1)$. This result is established in Lemma \ref{Lm:2.6} below.

Property $iii.)$, the standard property,  follows from abstract reasons (Theorem 5.6 of \cite{BGL},
which is applicable since we are dealing with representations of the principal or complementary series). 
A direct argument can also be given: 
For $\mathcal{O}=W_1$, the standard property is a general result about Tomita operators (see, e.g., \cite{RvD}). 
The first identity in \eqref{standard property} then follows 
by covariance and isotony (since every double cone is contined in some wedge).   

Adapting ideas of Borchers and Buchholz (see \cite[Lemma 3.2]{BoB}) to the one-particle space, 
one can show that the orthogonal complement of the complex linear span of~$\mathcal{H}_I$ is empty.
Hence, the second identity in \eqref{standard property} follows too.
\end{proof}

\begin{lemma}
\label{Lm:2.6} 
The multiplication of $\chi \in  C^\infty (S^1)$ on a vector of $\mathbb{H}^{1/2}(S^1)$ is a bounded  
operator in $\mathbb{H}^{1/2}(S^1)$.
\end{lemma}

\begin{proof} Let $g$ be a vector in $\mathbb{H}^{1/2}(S^1)$, \emph{i.e.}, 
$\sum_{k \in \mathbb{Z}} \widetilde{\omega} (k) | g_k |^2 < \infty $, where 	
		\begin{equation} \label{eqFourierCoeff}
		g_k := \frac{1}{\sqrt{2\pi r}}\int_{S^1} r {\rm d} \psi \;  {\rm e}^{i k \psi} f (\psi)   
        \end{equation}
are the Fourier coefficients~\eqref{eqFourierCoeff} of $g$. We have to show that there is a 
constant $c>0$ such that $\|\chi g\|_{\mathbb{H}^{1/2}(S^1)}\leq c\,\|g\|_{\mathbb{H}^{1/2}(S^1)}$. 
Note that  $\chi \in C^\infty (S^1)$ implies that, for all $\ell \in \mathbb{N}$,  
	\begin{equation}
	\label{star-star}
		(1+ |n|)^\ell \chi_k \to 0 \qquad \text{as} \qquad |k|
        \to \infty \; .  
	\end{equation}
In particular, $\chi \in \mathbb{H}^{1/2}(S^1)$. 
Now using $(\chi\cdot h)_k=\sum_{k'} \chi_{k-k'} h_{k'}$, one gets
	\begin{align*}
	  \| \chi h \|^2_{\mathbb{H}^{1/2}(S^1)} 
                 & = \sum_{k\in \mathbb{Z}} \Big|\widetilde\omega(k)^{1/2}
                 \sum_{k' \in \mathbb{Z}}\chi_{k-k'}  h_{k'} 	\Big|^2 	 
		 \\
		 & =   \sum_{k, k' \in \mathbb{Z}} \widetilde{\omega}
                 (k+k')  \, \bigl| \chi_{k} h_{k'}  \bigr|^2 \; . 		 
	\end{align*}
The facts that $\widetilde{\omega}$ is monotonically increasing and $\widetilde{\omega}(k)|k|^{-1}\to r^{-1}$ for large $k$ imply that there are positive constants\footnote{If (as we expect) 
the map $k \mapsto \widetilde{\omega}(k)$ is convex, than one can set $a=b=1$.} $a$, $b$ such that 
$|k| \leq b \widetilde\omega(k)$ and $\widetilde\omega(k) \leq \widetilde\omega(0) + a |k|$. 
Thus,
        \[
        \widetilde\omega(0) \le | \widetilde \omega(k+k') | \leq \widetilde\omega(0) +  a|k|+   a|k' |
        \qquad \forall k, k' \in \mathbb{Z} 
        \] 
and 
	\begin{align} 
	\label{eqTriangle} 
		\| \chi h \|^2_{\mathbb{H}^{1/2}(S^1)} 
		& \leq 
		\widetilde\omega(0) \| \chi \|^2_{L^2} \| h \|^2_{L^2}
		+ 
		a \| h \|^2_{L^2} \left( \sum_{k\in \mathbb{Z}} | k|   \,|\chi_{k} |^2 \right)  
		\nonumber
		\\
		& \qquad +
		a \| \chi \|^2_{L^2} \left( \sum_{k'\in \mathbb{Z}} | k'|   \, | h_{k'} |^2 \right)  
		\nonumber		
		\\ 
		& \leq 
		\left( (1+ab) \| \chi \|^2_{L^2} +  \frac{ab}{\widetilde\omega(0)} \|\chi  \|^2_{\mathbb{H}^{1/2}(S^1)}
		\right) \| h\|^2_{\mathbb{H}^{1/2}(S^1)} 	\; 	.
	\end{align}
\color{black}
In the second inequality we used $\widetilde \omega(0) \| \, . \|_{L^2} \le \| \, . \,   \|^2_{\mathbb{H}^{1/2}(S^1)}$.  
\end{proof}

\section{Nets of Local Algebras}
\label{sec:4}

The bosonic \emph{Fock space} $\mathcal{F}  = \mathbb{\Gamma}(\mathcal{H})$ over $\mathcal{H}$ is 
defined as the direct sum of the $n$-particle spaces:
	\[
		\mathbb{\Gamma}(\mathcal{H})  \doteq\oplus_{n= 0}^{\infty}  \;  \mathcal{H}^{\otimes_s^n} \;  , 
		\qquad  \mathcal{H}^{\otimes_s^0} \doteq \mathbb{C} \; ,  
	\]
with $\mathcal{H}^{\otimes_s^n}$ the n-fold totally symmetric tensor product $\otimes_s$ 
of~$\mathcal{H}$ with itself.  The {\em coherent vectors} 
	\[
		\mathbb{\Gamma}( h) \doteq \oplus_{n=0}^\infty \frac{1}{\sqrt{n!}} \underbrace{h \otimes_s \cdots \otimes_s h}_{n-times}  		  
	\]
form a total set in $\mathcal{F}$. The vector $\Omega_\circ= \mathbb{\Gamma}(0)$ is called the
Fock vacuum. One can also define\footnote{To the best of our knowledge, this formulation first 
appeared in \cite{Gui}. Our presentation follows closely \cite{Guido}.} second quantized operators: 

\begin{lemma}
Let $A$ be a closed, densely defined linear operator on $\mathcal{H}$ with domain~${\mathscr D} (A)$. Then
	\[
		\mathbb{\Gamma} (A) \colon  \mathcal{F} \to \mathcal{F} 
	\]
is the closure of the linear operator acting on the linear combinations of coherent vectors with 
exponent in~${\mathscr D} (A)$ such that
	\[
		\mathbb{\Gamma} (A) \mathbb{\Gamma} (h) = \mathbb{\Gamma} (A h) \; .
	\]
This exponentiation preserves self-adjointness, positivity and unitarity.  
\end{lemma}

For $h, g \in \mathcal{H}$, the relations
	\begin{align*} 
		V (h) V (g) 
		& = {\rm e}^{- i \, \Im \langle h ,  g \rangle } V (h+g) \; , 
		\\
		V (h) \Omega_\circ & ={\rm e}^{-\frac{1}{2} \| h \|^2}  \mathbb{\Gamma}( i  h) \; , 
	\end{align*}
define unitary operators, called the {\em Weyl operators}. They satisfy 
	\[ 
		V^*(h) = V (-h) 
		\quad \hbox{and} \quad V (0)= \mathbb{1} \; . 
	\]
We use the Weyl operators to associate a von Neumann algebras acting on the Fock space 
$\mathcal{F}$ to the wedge $W_1$: let~${\mathscr A}_\circ ( W_1 )$ denote the von Neumann 
algebra generated by the Weyl operators
	\[
		\bigl\{ V (h) \mid h \in \mathcal{H}(W_1) \bigr\} \; .  
	\]
Given a representation $U(\Lambda)$, $\Lambda \in SO_0 (1,2)$, of the
Lorentz group acting on Fock space which acts kinematically on the 
Cauchy surface $S^1$ (\emph{i.e.}, the rotations leaving $S^1$ invariant are 
given by the pull-back acting on one-particle wave functions), 
we can define von Neumann algebras associated to arbitrary
convex, causally complete,  bounded regions. We proceed in steps, repeating the ideas which lie
behind Definition~\ref{def-h-loc} and starting from the free field algebra  
of the wedge $W_1$. Note that $\Lambda W_1\subset
W_1\Leftrightarrow \Lambda=\Lambda_1(t)$ for some $t \in \mathbb{R}$.  

\begin{definition}
\label{locoinfield}
Given a unitary representation $\Lambda \mapsto U(\Lambda)$ of the
Lorentz group $SO_0(1,2)$ acting on Fock space $\mathcal{F}$ and satisfying the condition
	\begin{equation}
	\label{u-a-0}
		U \bigl( \Lambda_1(t) \bigr) {\mathscr A}_0 (W_1)
		U \bigl( \Lambda_1(t) \bigr)^{-1} = {\mathscr A}_0(W_1) \; , 
		\qquad t \in \mathbb{R} \; ,
	\end{equation}                 
we define the following von Neumann algebras:  
\begin{itemize}
\item[$ i.)$] For an arbitrary wedge $W= \Lambda W_1$, $\Lambda \in
  SO_0(1, 2)$, we set 
	 \begin{equation}
	 \label{a-w-U}
	 	{\mathscr A} ( W ) 
		\doteq U ( \Lambda) {\mathscr A}_\circ
		 \bigl(  W_1 \bigr) U ( \Lambda)^{-1} \;  . 	
	\end{equation}
\item[$ ii.)$] 
For an arbitrary bounded, causally complete, convex region (these are the de Sitter analogs of 
the double cones) ${\mathcal O} \subset \mathbb{dS}$, we set   
	\begin{equation} 
		\label{localfield}
		 {\mathscr A} ({\mathcal O}) 
		\doteq \bigcap_{W \supset {\mathcal O} } {\mathscr A}  \bigl( W \bigr) \; . 
	\end{equation}
\end{itemize}
The inclusion preserving map 
	\[
		{\mathcal O} \mapsto {\mathscr A} ({\mathcal O}) 
	\]
is called, in a slight abuse\footnote{Technically speaking, this family of algebras is not a ``net'', since 
in de Sitter space not every pair of double cones is contained in a double cone.} of the term,
the net of local von Neumann algebras for the bosonic  field on the de Sitter space $\mathbb{dS}$ 
transforming under $U$.   
\end{definition}

\goodbreak
\begin{remarks}
\label{IV.1} 
\quad
\begin{itemize}
\item[$i.)$] In case $U \equiv U_\circ$,  
	\[
		U_\circ (\Lambda) \doteq \mathbb{\Gamma} ( u ( \Lambda)) \; , \qquad 
		\Lambda \in SO_0 (1,2) \; , 
	\]
we will denote the generator of one-parameter unitary 
group $ t \mapsto U_\circ \bigl( \Lambda_1(t) \bigr)$ by~$L_\circ$
and the local algebra  by~${\mathscr A}_\circ ({\mathcal O})$. It follows from 
a result by Araki \cite[Theorem 1]{A1} and Proposition \ref{Prop-ii.4} $ii.)$ that 
	\[
		{\mathscr A}_\circ ({\mathcal O}_I) = \left\{ V(h)   \mid h \in \mathcal{H} , \;   
		{\rm supp\,} \left( \Re h,    \,  \omega^{-1}\Im h \right) \subset I \times I \right\} '' ; 
	\]
just as one might have expected. 
\item[$ii.)$] 
For the $\mathscr{P}(\varphi)_2$ model on the de Sitter space, 
the representation $U$ of $SO_0(1,2)$ in Definition~\ref{locoinfield}  
is generated by the rotations $U \bigl( R_0(\alpha) \bigr)$, $\alpha \in [0, 2 \pi)$, and the boosts 
$U \bigl( \Lambda_1(t) \bigr)$, $t \in \mathbb{R}$. The latter are defined as follows:  
\begin{itemize}
\item [$a.)$] The rotations  leaving the Cauchy surface $S^1$ invariant are the free ones, 
  \begin{equation} \label{eqRotFree}
    U \bigl( R_0(\alpha) \bigr) = \mathbb{\Gamma} (u(R_0(\alpha)) \; , \quad\alpha
    \in [0, 2 \pi) \; ; 
  \end{equation}
\item [$b.)$]  
The generator of the one-parameter unitary group $ t \mapsto U \bigl( \Lambda_1(t) \bigr)$
can be expressed in terms of canonical fields and canonical momenta (see \cite{BJM-1} for details): 
	\begin{equation} 
  	\label{L-int}
			\qquad 
			\; 
			L =  L_\circ + 
			\lim_{\epsilon \to 0}
			\int_{S^1} \,  r  \cos  \psi    \, {\rm d} \psi  \;   {:} \mathscr{P} \bigl(\varphi(\delta_\epsilon ( \, . \, -\psi))\bigr) {:} \;  ,  
	\end{equation} 
where $\mathscr{P}$ is a real-valued polynomial, bounded from below, $\varphi (h)$ is the generator of the 
one-parameter unitary group $s \mapsto V(sh)$, and $\delta_\epsilon$ approximates the Dirac 
delta function as $\epsilon \to 0$. As usual, the $\, {:} \;  {:} \, $  indicates normal ordering. 
\end{itemize}
\item [$iii.)$] The representations presented in $i.)$ and $ii.)$ extend to representations of $O(1,2)$ by adding the reflections 
$\mathbb{\Gamma} ( u ( P_1))$ and $\mathbb{\Gamma} (u(T))$. Note that in the interacting case $ii.)$ the reflection 
	\[
		U (\Theta_W) \doteq   U (\Lambda)\mathbb{\Gamma} \bigl( u(\Theta_W) \bigr)U^{-1}(\Lambda) 
	\]
associated to a wedge $W = \Lambda W_1$, $\Lambda \in SO_0(1,2)$, whose edges do \emph{not} lie on the Cauchy surface $S^1$
will in general \emph{not} coincide with $\mathbb{\Gamma} (u(\Theta_W))$.
\end{itemize}
\end{remarks}

In the general case, we will need a criterium for the representation of the Lorentz group 
which ensures that the intersection in \eqref{localfield} is not trivial. 
A sufficient condition is the following: 

\begin{definition}
\label{def:4.2}
Assume the von Neumann algebras are defined as in Definition \ref{locoinfield}.
The net of local algebras is said to satisfy \emph{finite speed of light}, if
for any wedge $W$, the algebra $\mathscr{A} (W) $ is contained in the time-zero Weyl algebra
	\begin{equation}
		\left\{ V(h)  \mid  h \in \mathcal{H} \, ; \;  
		\operatorname{supp} \Re h \subset J \, , \; 
		\operatorname{supp} \omega^{-1}\Im h  \subset J  \right\} '' ,
		\label{e6.1ef}
	\end{equation}
where $J = \Gamma (W) \cap S^1$.
\end{definition}

\begin{remarks}
\quad
\begin{itemize}
\item[$i.)$] 
If we associate von Neumann algebras to intervals $I$ by setting
	\[
		\mathscr{R}( I ) \doteq \mathscr{A} \bigl( {\mathcal O}_I \bigr) \; , 
		\qquad \mathscr{R}_\circ ( I ) \doteq \mathscr{A}_\circ \bigl( {\mathcal O}_I \bigr) \; , 
		\qquad I \subset S^1 \; , 
	\]
then \eqref{e6.1ef} may be formulated in the same way as finite speed of light
was originally defined by Glimm and Jaffe in \cite[see Theorem 6.7 for the free case and 
Theorem 8.1 for the interacting case]{GJ}, namely by requesting that 
	\[
		U (\Lambda) \mathscr{R}(I)U^{-1} (\Lambda) \subset {\mathscr{R}_\circ} \bigl(  \Gamma (\Lambda I ) \cap S^1 \bigr)\; . 
	\]
Of course, on Minkowski space one would use a time translations by some $|t|< \delta$ instead of the boosts, 
and then $\Gamma (\Lambda I ) \cap S^1$ would just be equal to $I + [-\delta , \delta]$. 
\item[$ii.)$] 
For the $\mathscr{P}(\varphi)_2$ model on de Sitter space, property~\eqref{e6.1ef}, 
which encodes finite speed of light,  was established in Theorem 10.1.1 in \cite{BJM-1}. 
As in \cite[Theorem 8.1]{GJ} the key property, which implies finite speed of light, is the additivity 
of the integral in the second term in \eqref{L-int}.
\end{itemize}
\end{remarks}

We can now state a key result of our investigation. 

\begin{theorem} \label{AA0}
Assume the net of local algebras satisfies \emph{finite speed of light}  (in the sense of 
Definition \ref{def:4.2}). Then the local algebras associated to  an interval $I \subset S^1$ on the Cauchy surface
coincide with those of the free theory, \emph{i.e.},
	\[
		\mathscr{A} ({\mathcal O}_I) = \mathscr{A}_\circ ({\mathcal O}_I) \; , 
		\qquad I \subset S^1 \; . 
	\]
\end{theorem}

\begin{proof} The proof follows the ideas exposed in the proof of Proposition \ref{Prop-ii.4} $ii.)$. Thus the key step is to 
show that for any wedge $W$ which contains ${\mathcal O}_I$, we have 
	\[
		\mathscr{A}  (W')  \subset \left\{ V(h) \mid h \in \mathcal{H} , \; 
		\operatorname{supp} \Re h \subset I^c \, , \; 
		\operatorname{supp}  \omega^{-1}\Im h \subset I^c  \right\} '' ,
	\]
where $I^c \doteq S^1 \setminus \overline{I} $. As the edges of $W$ are necessarily space- or light-like 
to~${\mathcal O}_I$, this inclusion follows from finite speed of light as expressed in \eqref{e6.1ef}. By duality, 
	\[
		\mathscr{A}  (W)  \supset \underbrace{\left\{ V(h)  \mid h \in \mathcal{H} , \;  
		\operatorname{supp} \Re h \subset I \, , \; 
		\operatorname{supp} \omega^{-1}\Im h \subset I \color{black} \right\} ''}_{ = \mathscr{A}_\circ ( {\mathcal O}_I) } \, , 
	\]
whenever $W$ includes ${\mathcal O}_I$. 
\end{proof}

\goodbreak

\begin{remark}
The circle $S^1$, which we use to identify the free field and the interacting field, could be replaced by 
any space-like geodesic $\Lambda S^1$, $\Lambda \in SO_0(1,2)$. The Fock space simply carries two (in fact, 
infinitely many if one just varies the coupling constants) nets of local algebras, namely 
$\mathcal{O} \mapsto {\mathcal A}_\circ (\mathcal{O})$ and $\mathcal{O} \mapsto {\mathcal A} (\mathcal{O})$, 
and one may identify them on any of the space-like geodesic $\Lambda S^1$, $\Lambda \in SO_0(1,2)$.
\end{remark}

\section{The Haag-Kastler Axioms}

Assume we have defined a net $ {\mathcal O} \mapsto \mathscr{A} ({\mathcal O}) $ of local algebras 
according to Definition \ref{locoinfield}, with the rotations $R_0(\alpha)$, leaving the Cauchy surface $S^1$ invariant, 
implemented by the kinematic representation $ U \bigl( R_0(\alpha) \bigr) 
= \mathbb{\Gamma} (u(R_0(\alpha))$, $\alpha \in [0, 2 \pi)$,  
and the boosts $\Lambda_1(t)$, leaving the wedge $W_1$ invariant, implemented by the modular group for the pair
$\bigl({\mathcal A}_\circ (W_1), \Omega\bigr)$, with $\Omega$ a cyclic and separating for~${\mathcal A}_\circ (W_1)$.

\begin{theorem} 
\label{th:6}
In case the net $ {\mathcal O} \mapsto \mathscr{A} ({\mathcal O}) $ respects  \emph{finite speed of light} 
in the sense of Definition~\ref{def:4.2},  it will also satisfy the following Haag-Kastler axioms:
\begin{enumerate}
\item [$i.)$]  {\em (Isotony).} The local algebras satisfy 
	\[
		\mathscr{A} ({\mathcal O}_1) \subset \mathscr{A} ({\mathcal O}_2)  
		\quad \hbox{if} \quad {\mathcal O}_1 \subset {\mathcal O}_2  \; . 
	\]
Here ${\mathcal O}_1$ and ${\mathcal O}_2$ are either double cones or wedges (but the result extends to
arbitrary regions once \eqref{Additivity} has been established). 
\item [$ii.)$]  {\em (Locality).} The local algebras satisfy 
	\[
		\mathscr{A} ({\mathcal O}_1) \subset \mathscr{A} ({\mathcal O}_2)' 
		\quad \hbox{if} \quad {\mathcal O}_1 \subset {\mathcal O}_2 ' \; . 
	\]
Here ${\mathcal O}'$ denotes the space-like complement of 
${\mathcal O}$ in $\mathbb{dS}$ and $\mathscr{A} ({\mathcal O})'$ is the commutant 
of $\mathscr{A} ({\mathcal O})$ in $\mathscr{B}(\mathcal{F})$. 
\item [$iii.)$] {\em (Covariance).}
The representation $U \colon \Lambda \mapsto U ( \Lambda) $ acts geometrically, \emph{i.e.}, 
	\[
		U ( \Lambda)  \mathscr{A} ({\mathcal O}) U ( \Lambda)^{-1} 
		= \mathscr{A} (\Lambda {\mathcal O} )  \; , \qquad \Lambda  \in SO_0 (1,2) \; .  
	\]
\item [$iv.)$] {\em (Existence and Uniqueness of the Vacuum \cite{BoB}).} 
There exists a unique (up to a phase\footnote{The phase is uniquely fixed, if one insists 
that the vector $\Omega$ lies in the natural positive 
cone~$\mathcal{P}^\natural (\mathcal{A}_\circ (W_1), \Omega_\circ)$.}) 
unit vector in $\mathcal{F}$, namely $\Omega$, which 
\begin{enumerate}
\item [a.)] is  invariant under the action of $U  (SO_0(1,2))$; 
\item [b.)] satisfies the geodesic KMS condition: for every wedge $W = \Lambda W_1$, 
$\Lambda \in SO_0(1,2)$, the partial state
	\[ 
		\qquad
		\qquad 
		\omega_{\upharpoonright \mathscr{A}(W)} (A) \doteq \langle \Omega, A  \Omega \rangle \; , 
		\quad A \in  \mathscr{A} (W) \; ,
	\]
satisfies the KMS-condition at inverse temperature~$\beta = 2 \pi r $ with respect to the one-parameter 
group $t \mapsto U (\Lambda_{W}( t/r)) $,  $t \in {\mathbb R} $.
\end{enumerate}
\item[$v.)$]  {\em (Additivity).}  
For $X$ a double cone or a wedge, there holds
	\begin{equation} 
		\label{Additivity} 
		\mathscr{A}(X) = \bigvee_{{\mathcal O\subset X}} \mathscr{A}({\mathcal O}) \; .
	\end{equation}
The right hand side denotes the von Neumann algebra generated by the 
local algebras associated to double cones ${\mathcal O}$ contained in $X$.
(It thus makes sense to define~$\mathscr{A}(X)$ for arbitrary regions
$X$ by Equ.~\eqref{Additivity}.) 
\item [$v'.)$] {\em (Weak additivity).}
For each double cone ${\mathcal O} \subset \mathbb{dS}$ there holds 
	\[  
 	 	\bigvee_{\Lambda \in SO_0 (1,2) } \mathscr{A}  (\Lambda
                {\mathcal O}) = \mathscr{A}  (\mathbb{dS})  
                  \quad \bigl( = \mathscr{B}({\mathcal F}) \bigr) \; . 
	\] 
\item [$vi.)$] {\em(Time-slice axiom \cite{CF}).}  
Let $I$ be an interval on a geodesic Cauchy surface and let $I''$ be its causal completion.
Let $\Xi \subset I'' $ be a neighbourhood of $I$. Then 
	\[
		\mathscr{A}( \Xi ) = \mathscr{A}(I'') \; ,  
	\] 
where both algebras are defined via Eq.~\eqref{Additivity}. In particular, the algebra of observables located within an 
arbitrary small time--slice coincides  with the algebra of all observables.
\end{enumerate}
\end{theorem}

\begin{proof} 
Property $i.)$, \emph{isotony}, follows directly from the definition given in \eqref{localfield}.
Next, let us establish property $ii.)$, \emph{locality}.
If $\mathcal{O}_1$ and~$\mathcal{O}_2$ are two
space-like separated causally complete, open and bounded regions,
then there exists a wedge  $W = \Lambda W_1$ such that
	\[
		\mathcal{O}_1 \subset W 	\quad \text{and} \quad \mathcal{O}_2 \subset W' \, . 	
	\]
Now the interacting net inherits wedge duality
        \[
        		\mathscr{A} (W')=\mathscr{A}(W)'
        \]
from the identity $\mathscr{A}_\circ (W_1')=\mathscr{A}_\circ (W_1)'$
using \eqref{a-w-U}.   
These facts imply locality. 

Now, let us prove property $iii.)$, \emph{covariance}. 
Let $\Lambda \in SO_0(1,2)$ be fixed. By construction, the set of all wedges equals
$\{ \Lambda W_1 \mid \Lambda \in SO_0(1, 2) \}$. Thus 
	\begin{align*}
	{\mathscr A} (\Lambda {\mathcal O} ) 
	& =  \bigcap_{\Lambda {\mathcal O} \subset  \Lambda W} 
	{\mathscr A} ( \Lambda W)
	 =   \bigcap_{ {\mathcal O} \subset W} 
	U ( \Lambda) {\mathscr A} ( W) U ( \Lambda)^{-1} 
	\\
	& =  U ( \Lambda)
	\Bigl( \; \bigcap_{ {\mathcal O} \subset W}  {\mathscr A} (W) \Bigr) U ( \Lambda)^{-1}
	=  U ( \Lambda)  {\mathscr A} ( {\mathcal O} )   U ( \Lambda)^{-1} \; ,
	\end{align*}
proving covariance. 

Property  $iv.)$ is established next: \emph{existence} of the de Sitter vacuum is guaranteed by 
construction, as the state induced by the vector $\Omega$ is a thermal state for the Hawking temperature $(2\pi r)^{-1}$ 
with respect to modular group for the pair $\bigl({\mathcal A}_\circ (W_1), \Omega\bigr)$. By covariance, this property extends 
to arbitrary wedges. \emph{Uniqueness of the de Sitter vacuum state}, 
was established in \cite[Theorem 9.3.7]{BJM-1}. 
It follows from the fact that for the free field the KMS state for the algebra 
on~${\mathscr A}_\circ \bigl(  W_1 \bigr)$ is unique, hence ${\mathscr A} \bigl(  W_1 \bigr) = {\mathscr A}_\circ \bigl(  W_1 \bigr)$
(see Definition \ref{locoinfield}) is a factor, and uniqueness of the interacting KMS state now is 
a direct consequence of \cite[Proposition~5.3.29]{BR}. 

Next, we will establish property $v.)$, \emph{additivity}.
The inclusion 
	\[ 
		\mathscr{A}(X) \supset \bigvee_{{\mathcal
    			O\subset X}} \mathscr{A}({\mathcal O}) \; .
	\]
is a consequence of isotony. Moreover, if $X$ is a double cone, then 
$X$ itself is among the double cones on the right hand side, so the inclusion $\subset$  
automatically holds.  It remains to prove the inclusion $\subset$ if $X$ is a wedge.
If $X=W_1$, then $\mathscr{A}(X)$ coincides with $\mathscr{A}_\circ(X)$, for which the additivity 
property for the one-particle space \eqref{additivity} implies
	\[
		\mathscr{A}_\circ(X) =  \bigvee_{R_0 (\alpha) I\subset I_+}\mathscr{A}_\circ \bigl(R_0 (\alpha) I \bigr) \; , 
		\qquad \alpha \in [0, 2\pi) \; , 
	\]
where $I$ is any (arbitrarily small) interval contained in $I_+\doteq
W_1\cap S^1$, whose causal completion ${\mathcal O}_I=I''$.  Thus,
	\[
		\mathscr{A} \bigl(W_1 \bigr) 
		= \bigvee_{R_0 (\alpha) I \subset I_+}\mathscr{A} \bigl(R_0 (\alpha) {\mathcal O}_I \bigr)
			\subset \bigvee_{{\mathcal O}\subset X}\mathscr{A}({\mathcal O}) \; .        
	\]
Thus, for $X=W_1$ the inclusion $\subset$ holds. By covariance, it
also holds if $X$ is any other wedge. 
\color{black}

Property $v'.)$, \emph{weak additivity}, follows form a similar argument: for each double cone
${\mathcal O} \subset \mathbb{dS}$ there exists a Lorentz transformation
$\Lambda_0 \in SO_0(1,2)$ such that ${\mathcal O}= \Lambda_0  \mathcal{O}_I$ 
for some open interval $I \subset S^1$. Now
	\begin{align*}
 	 	\bigvee_{\Lambda \in SO_0(2,1)} \mathscr{A}  \bigl( \Lambda {\mathcal O} \bigr) 
                & = \bigvee_{\Lambda \in SO_0(2,1)} \mathscr{A}  \bigl( \Lambda \Lambda_0 {\mathcal O}_I \bigr) \\
		& = \bigvee_{\Lambda \in SO_0(2,1)} \mathscr{A}  \bigl(\Lambda  {\mathcal O}_I \bigr)
		\supset \bigvee_{ \alpha \in [0, 2\pi) } \mathscr{A}  \bigl(R_0 (\alpha) {\mathcal  O}_I \bigr) \\
		&= \bigvee_{\alpha \in [0, 2\pi) } \mathscr{A}_\circ  \bigl(R_0 (\alpha) {\mathcal  O}_I \bigr) 
		= \mathscr{B}({\mathcal F})  \; . 
	\end{align*} 
Again, the last equality relies on the additivity property for the one-particle space. Hence, property $v'.)$, 
\emph{weak additivity}, is established. 

Let us prove property $vi.)$, the \emph{time-slice axiom}. As the geodesic Cauchy surfaces are precisely 
the Lorentz transforms of intervals on the equator $S^1$  (see, \emph{e.g.},~\cite{ONeill}), it 
is sufficient to consider an interval $I \subset S^1$ and a neighbourhood $\Xi \subset I''$
of $I$. If the length of $I$ is less than $\pi r$, then $I''$ is a double cone and ${\mathscr A}(I'')={\mathscr A}_\circ(I'')$. 
Pick any double cone $\mathcal{O}\subset \Xi$ with base on $S^1$. Then the additivity property  
of the free net implies that 
	\[ 
		{\mathscr A}_\circ(I'') = 
		\bigvee_{R_0 (\alpha) \mathcal{O}\subset \Xi }
   		{\mathscr A}_\circ(R_0 (\alpha) \mathcal{O}) \; , \qquad \alpha \in [0, 2\pi) \; .  
	\]
Now ${\mathscr A}_\circ \bigl( R_0 (\alpha) \mathcal{O} \bigr)$ coincides with 
${\mathscr A}\bigl( R_0 (\alpha) \mathcal{O} \bigr)$, and hence the above identity implies 
${\mathscr A}(I'')\subset {\mathscr A}( \Xi )$. The other inclusion follows from isotony. 

\goodbreak
If the length of $I$ is at least $\pi r$, then additivity implies that 
${\mathscr A}(I'')$ is generated by the wedge algebras ${\mathscr A}(R_0(\alpha) W_1)
= {\mathscr A}_\circ(R_0(\alpha) W_1)$, $R_0(\alpha) W_1 \subset I''$. 
Hence, as before, ${\mathscr A}(I'')={\mathscr A}_\circ(I'')$.
As before, the additivity property  of the free net implies that 
	\[ 
		{\mathscr A}_\circ(I'') = 
		\bigvee_{R_0 (\alpha) \mathcal{O}\subset \Xi }
   		{\mathscr A}_\circ(R_0 (\alpha) \mathcal{O}) \; , \qquad \alpha \in [0, 2\pi) \; .  
	\]
Now, just as before,  ${\mathscr A}_\circ \bigl( R_0 (\alpha) \mathcal{O} \bigr)$ coincides with 
${\mathscr A}\bigl( R_0 (\alpha) \mathcal{O} \bigr)$, 
and hence the above identity implies 
${\mathscr A}(I'')\subset {\mathscr A}( \Xi )$. The other inclusion follows from isotony. 
\color{black}
\end{proof}

The local algebras ${\mathscr A}_\circ ({\mathcal O})$ and ${\mathscr A}_\circ ({\mathcal W})$, 
for double cones and wedges, respectively, 
are hyper-finite type III$_1$ factors; see~\cite{FG}.  

\begin{corollary}
Let $X$ be either a double cone ${\mathcal O}$ or a wedge $W$. It follows that
${\mathscr A} ( X )$ is a hyperfinite type $III_1$ factor. 
\end{corollary}

\begin{proof} 
For wedges and double cones with base in $S^1$ the interacting algebras coincide with 
the free ones, for which the factor property is known as stated before the Corollary. 
Any other wedge or double cone can be mapped by a Lorentz transformation to 
one with base in $S^1$, and the claim follows by covariance.
\end{proof}

\begin{remarks}
\quad
\begin{itemize}
\item [$i.)$] 
Finite speed of light was used to establish the additivity properties $v.)$, $v'.)$ as well as $vi.)$.
The properties $i.)$ to $iv.)$ did not require this property of the representation~$U$ of $SO_0(1,2)$. 
 
\item [$ii.)$]
It has been shown by Borchers and Buchholz that if one assumes the 
geodesic KMS condition to hold for some $\beta > 0$, then~automatically $\beta = 2 \pi \, r $; see
\cite[Theorem~6.2]{BoB}.
\end{itemize}
\end{remarks}

\begin{theorem}[Borchers \& Buchholz \cite{BoB}]
\label{BOBU}
The following properties hold for any algebraic quantum theory ${\mathcal O} \to {\mathscr A}({\mathcal O})$ on the de Sitter space, 
which satisfies the Haag-Kastler axioms $i.)-v.)$ stated in Theorem \ref{th:6}:  
\begin{itemize}
\item[$ i.)$] {\em (The Reeh-Schlieder property).}
For any open region ${\mathcal O} \subset \mathbb{dS}$ there holds
\[ \overline{ {\mathscr A} ({\mathcal O}) \Omega } = \mathfrak{h}  \; . \]
\item[$ ii.)$] {\rm (Wedge duality).} For any wedge $W \subset \mathbb{dS}$ there holds 
\[ J_{W}  \; {\mathscr A}   (W  ) \, J_{W}   
= {\mathscr A}  (W  )' = {\mathscr A}  ({W }' ) \; . \]
Here $J_{W}$ denotes the modular conjugation associated to $({\mathscr A}  ({W } ) , \Omega)$.  

\smallskip

\item[$ iii.)$] {\em (The PCT symmetry).}
Let $T$ denote the time reflection and $P_1$ the reflection across the $x_0$-$x_1$-plane. 
Then the modular conjugation $J_{W_1}$ for the wedge~$W_1$ is an anti-unitary 
representer of the reflection 
	\[
	TP_1 \in O (1, 2) \; , 
	\]
which induces the corresponding action on 
$U (SO_0 (1,2))$ and ${\mathscr A} (\mathbb{dS})$. 
A similar result holds for an arbitrary wedge~$\Lambda W_1$, $\Lambda \in SO_0(1,2)$; the relevant reflection is then 
$ \Lambda TP_1 \Lambda^{-1} \in O (1, 2) $. 
\end{itemize}
\end{theorem}

\begin{proof}
$ i.)$ is Theorem 3.1 in  \cite{BoB}; $ ii.)$ is Proposition 6.1 in \cite{BoB}; $ iii.)$ is Theorem~6.3 in \cite{BoB}.
\end{proof}

\begin{remark}
Another result one might want to mention is that uniqueness of the interacting de Sitter vacuum 
implies that it is {\em weakly mixing} with respect to the action 
of the boosts \cite[Corollary~4.4]{BoB}. 
\end{remark}

\section{Summary}
\label{sec:5}

In a truly pioneering work \cite{FHN}, published in 1975, Figari, Høegh-Krohn and Nappi constructed the 
Wightman $n$-point functions of the ${\mathscr P} (\varphi)_2$ model in a wedge 
of de Sitter space-time from the Schwinger functions on the Euclidean sphere. More recently, the authors 
provided an elementary construction (which ignores domain questions\footnote{It was shown
in  \cite{BJM-1} that the operator sum in \eqref{L-int}, which involves two operators that are neither bounded 
from below nor from above, is essentially self-adjoint on its natural domain.}) of this model in its canonical 
formulation \cite{JM}. The main objective of this work was to show that once the ${\mathscr P} (\varphi)_2$ model on 
the de Sitter model has been established in its canonical formulation, a covariant formulation is readily accessible using 
the machinery developed by Brunetti, Guido and Longo \cite{BGL}. As the interacting representation of the Lorentz 
group  satisfies finite speed of light, the Haag-Kastler type axioms of Borchers and Buchholz \cite{BoB} can easily 
be verified. To the best of our knowledge, this is the first time that a \emph{non-perturbative, covariant, interacting} 
quantum field theory satisfying the basic properties encoded in the axioms has been established on a curved space-time.  

\paragraph*{Acknowledgments.}
We thank two anonymous referees for stimulating substantial improvements of the manuscript.
Jens Mund was partially supported by the São Paulo Research Foundation
(FAPESP), grant number 2014/24522-9,  and both authors were supported by the Conselho Nacional de
Desenvolvimento Cientifico e Tecnológico (CNPq). Jens Mund is also grateful to FAPEMIG, CAPES and Finep.  
Both authors would like to thank Jo\~ao Carlos Alves Barata for his contribution to this work.


\begin{thebibliography}{99.}
\bibitem{A1} 
	H.~Araki,  
	\textit{A lattice of von Neumann algebras associated with the quantum theory of a free Bose field}, 
	{J. Math. Phys. \textbf{4}}
	{(1963)}
	{1343--1362}.
\bibitem{A2} 
	H.~Araki,  
	\textit{Von Neumann algebras of local observables for free scalar field}, 
	{J. Math. Phys. \textbf{5}}
	{(1964)}
	{1--13}.
\bibitem{BJM-1} 
	{J.~Barata, C.\ J\"akel and J. Mund},
	\textit{Interacting quantum fields on de Sitter Space}, 
	{see 	arXiv:1607.02265v2}. 
\bibitem{BoB} 
	{H.-J. Borchers and D. Buchholz}, 
	\textit{Global properties of the vacuum states in de Sitter space}, 
	{Ann.~Inst.~H.~Poincar\'e~\textbf{A70}}
	{(1999)}
	{23--40}.
\bibitem{BiWia} 
	J.J. Bisognano and E.H. Wichman, 
	\textit{On the duality condition for a hermitian scalar field}, 
	{J. Math. Phys.~\textbf{16}} 
	{(1975)}
	{985--1007}.
\bibitem{BiWib} 
	J.J. Bisognano and E.H. Wichman,
	\textit{On the duality condition for quantum fields}, 
	{J. Math. Phys.~\textbf{17}}
	{(1976)} 
	{303--321}.
\bibitem{BR} 
	O. Bratteli and D.W. Robinson, 
	\textit{Operator Algebras and Quantum Statistical Mechanics~I, II}, 
	{Sprin\-ger-Verlag, New York-Heidelberg-Berlin},
	{1981}.
\bibitem{BM} 
	{J. Bros and U. Moschella}, 
	\textit{Two-point functions and quantum fields in de Sitter universe}, 
	{Rev.\ Math.\ Phys.~\textbf{8}}
	{(1996)}
	{327--391}.
\bibitem{BGL}
	R. Brunetti, D. Guido and R. Longo, 
	\textit{Modular localization and Wigner particles}, 
	Rev. Math. Phys.~\textbf{14} 
	{(2002)}
	{759--785}.
\bibitem{CF} 
	{B. Chilian  and K. Fredenhagen},  
	\textit{The time slice axiom in perturbative quantum field theory on globally hyperbolic spacetimes},  
	{Comm. Math. Phys.~\textbf{287}}
	{(2009)} 
	{513--522}.
\bibitem{DJP} 
	J. Derezinski, V. Jaksic and C.-A. Pillet, 
	\textit{Perturbation theory of $W^*$-dynamics, Liouvilleans and KMS states}, 
	{Rev. Math. Phys. \textbf{15}} 
	{(2003)} 
	{447--489}.
\bibitem{FHN} 
	R. Figari, R. H\o egh-Krohn  and C.R. Nappi, 
	\textit{Interacting relativistic boson fields in the de Sitter universe with two space-time dimensions}, 
	{Comm. Math. Phys. \textbf{44}}
	{(1975)}
	{265--278}.
\bibitem{FG} 
	{F. Figliolini and D. Guido}, 
	\textit{The Tomita operator for the free scalar field}, 
	{Ann. H. Poincar\'e~\textbf{51}}
	{(1989)}
	{419--435}.
\bibitem{GJ} 
	J. Glimm and A. Jaffe, \textit{Boson Quantum Field models}, published in
	\textit{Collected Papers, Vol. I: Quantum Field Theory and Statistical Mechanics; Expositions; 
	Vol. II: Constructive Quantum Field Theory, Selected Papers},
	{Birkh\"auser, Boston},
	{1985}.
\bibitem{Gui}
	{A. Guichardet},
	\textit{Tensor products of $C^*$-Algebras (Infinite)},
	{Aarhus University},
	{Lecture Notes Series \textbf{13}},
	{1969}.
\bibitem{Guido} 
	D. Guido,  
	\textit{Modular theory for the von Neumann algebras of local quantum physics}, 
	{Cont. Math. \textbf{534}} 
	(2011)
	97--120. 
\bibitem{JM} 
	{C.\ J\"akel and J. Mund},
	\textit{Canonical interacting quantum fields on two-dimensional de Sitter space}, 
	{Phys. Lett. \textbf{B 772}}
	{(2017)}
	{786--790}.
\bibitem{Le2} 
	{G. Lechner},  
	\textit{Construction of quantum field theories with factorizing $S$-matrices}, 
	{Comm. Math. Phys.~\textbf{277}}
	{(2008)}
	{821--860}.
\bibitem{Le3} 
	{G. Lechner},  
	\textit{On the existence of local observables in theories with a factorizing S-matrix}, 
	{J. Phys. A~\textbf{38}} 
	{(2005)}	
	{3045--3056}.
\bibitem{LL} 
	{G. Lechner and R. Longo}
	\textit{Localization in Nets of Standard Spaces}
	{Comm. Math. Phys. \textbf{336}}
	{(2015)}
	{27--61}.
\bibitem{Masuda}
	K. Masuda, K. (1972). 
	\textit{Anti-locality of the one-half power of elliptic differential operators}, 
	RIMS~\textbf{8}
	{207--210}. 
\bibitem{ONeill}
	{B. O'Neill, Semi-Riemannian Geometry With Applications to Relativity}, 
	Academic Press,
	1983.
\bibitem{RvD} 
	M.~A.\ Rieffel and A. Van Daele, 
  	\textit{A bounded operator approach to {T}omita-{T}akesaki theory},
  	Pacific J. Math.~\textbf{69} 
	(1977) 
	187--221.
\bibitem{S97} 
	B. Schroer, 
	\textit{Modular localization and the Bootstrap-formfactor program},
	Nucl.\ Phys. \textbf{B 499} 
	(1997) 
	547--568.
\bibitem{Verch}
	R. Verch, 
	\textit{Antilocality and a Reeh-Schlieder theorem on manifolds}, 
	Lett. Math. Phys.~\textbf{28}, 
	{(1993)}
	{143--154}. 
\bibitem{Weiner} 
	M. Weiner
	\textit{An algebraic version of Haag's theorem}, 
	{Comm. Math. Phys. \textbf{305}}
	{(2011)}
	{469--485}.
\end{thebibliography}
\end{document}